\documentclass{amsart}
\title[On generalizations of the pentagram map]{On generalizations of the pentagram map: discretizations of AGD flows}
\author{Gloria Mar\'i Beffa}
\usepackage{amsmath,amsfonts,amsthm}
\usepackage{graphicx}

\newtheorem{theorem}{Theorem}
\numberwithin{theorem}{section}

\newtheorem{conjecture}[theorem]{Conjecture}
\newtheorem{lemma}[theorem]{Lemma}

\newtheorem{proposition}[theorem]{Proposition}

\theoremstyle{definition}
\newtheorem{definition}[theorem]{Definition}

\def\RP{\mathbb {RP}}
\def\R{\mathbb R}
\def\sl{\mathfrak{sl}}
\def\g{\mathfrak{g}}
\def\bal{\mathfrak{b}}

\def\HH{\mathcal{H}}

\def\C{\mathcal{C}}
\def\PSL{\mathrm{PSL}}
\def\SL{\mathrm{SL}}
\def\M{\mathcal{M}}
\def\N{\mathcal{N}}

\def\P{\mathcal{P}}
\def\n{\mathfrak{n}}
\def\l{\Gamma}
\def\le{\Gamma_\epsilon}
\def\e{\epsilon}
\def\al{\alpha}
\def\be{\beta}
\def\D{\Delta}
\def\p{\Pi}
\def\s{\varsigma}
\def\Lo{\mathcal{L}}

\def\r{{\bf r}}
\def\k{{\bf k}}
\def\ka{\kappa}

\def\h{{\mathfrak h}}
\def\a{{\bf a}}
\def\b{{\bf b}}
\def\c{{\bf c}}
\def\m{{\bf m}}
\def\nb{{\bf n}}
\def\g{\gamma}
\def\d{\delta}

\def\M{\mathbf{M}}

\begin{document}
\maketitle
\begin{abstract} In this paper we investigate discretizations of AGD flows whose projective realizations are defined by intersecting different types of subspaces in $\RP^m$. These maps are natural candidates to generalize the pentagram map,  itself defined as the intersection of consecutive shortest diagonals of a convex polygon, and a completely integrable discretization of the Boussinesq equation. We conjecture that the $k$-AGD flow in $m$ dimensions can be discretized using one $k-1$ subspace and $k-1$ different $m-1$-dimensional subspaces of $\RP^m$.
\end{abstract}
\section{Introduction}
The pentagram map is defined on planar, convex n-gons, a space we will denote by $\C_n$. The map $T$ takes a vertex $x_n$ to the intersection of two segments: one is created by joining the vertices to the right and to the left of the original one, $\overline{x_{n-1}x_{n+1}}$, the second one by joining the original vertex to the second vertex to its right $\overline{x_nx_{n+2}}$ (see Fig. 1). These newly found vertices form a new n-gon. The pentagram map takes the first n-gon to this newly formed one. As surprisingly simple as this map is, it has an astonishingly large number of properties, see \cite{S1, S2, S3}, \cite{ST}, \cite{OST} for a thorough description.

The name pentagram map comes from the following classical fact: if $P \in \C_5$ is a pentagon, then $T(P)$ is projectively equivalent to $P$. Other relations seem to be also classical: if $P\in \C_6$ is a hexagon, then $T^2(P)$ is also projectively equivalent to $P$. The constructions performed to define the pentagram map can be equally carried out in the projective plane. In that case $T:\C_5\to \C_5$ is the identity, while $T:\C_6\to \C_6$ is an involution. In general, one should not expect to obtain a closed orbit for any $\C_n$; in fact orbits seem to exhibit a quasi-periodic behavior classically associated to completely integrable systems. This was conjectured in \cite{S3}.

A recent number of papers (\cite{OST, S1, S2, S3, ST}) have studied the pentagram map and stablished its completely integrable nature, in the Arnold-Liouville sense. The authors of \cite{OST} defined the pentagram map on what they called  {\it n-twisted polygons}, that is, infinite polygons with vertices $x_m$, for which $x_{n+m} = M(x_m)$, where $M$ is the {\it monodromy}, a projective automorphism of $\RP^2$. They proved that, when written in terms of the {\it projective invariants} of twisted polygons, the pentagram map was in fact Hamiltonian and completely integrable. They displayed a set of preserved quantities and proved that almost every universally convex n-gon lie on a smooth torus with a $T$-invariant affine structure, implying that almost all the orbits have a quasi-periodic motion under the map. Perhaps more relevant to this paper, the authors showed that the pentagram map, when expressed in terms of projective invariants, is a discretization of the Boussinesq equation, a well-known  completely integrable system PDE.

The Boussinesq equation is one of the best known completely integrable PDEs. It is one of the simplest among the so-called Ader--Gel'fand-Dickey (AGD) flows. These flows are biHamiltonian and completely integrable. Their first Hamiltonian structure was originally defined by Adler in \cite{A} and proved to be Poisson by Gel'fand and Dickey in \cite{GD}. The structure itself is defined on the space $\Lo$ of periodic and scalar differential operators of the form
 \[
L = D^{n+1} + k_{n-1}D^{n-1}+\dots + k_1 D+k_0,
\]
Hamiltonian functionals on $\Lo$ can be written as
 \[
 \HH_R(L) = \int_{S^1} \mathrm{res}(RL) dx
 \]
 form some pseudo-differential operator $R$, where res denotes the residue, i.e., the coefficient of $D^{-1}$. The Hamiltonian functionals for the AGD flows are given by 
 \begin{equation}\label{H}
 \HH(L) = \int_{S^1}\mathrm{res}(L^{k/m+1})dx
 \end{equation}
 $k=1,2,3,4,\dots, m$. The simplest case $m=1$, $k=1$ is equal to the KdV equation, the case $m = 2$ and $k = 2$ corresponds to the Boussinesq equation. Some authors give the name AGD flow to the choice $k = m$ only. 
 
The author of this paper linked the first AGD Hamiltonian structure and the AGD flows to projective geometry, first in \cite{HLM} and \cite{M4}, and later in \cite{M1} and \cite{M3} where she stablished a geometric connection between Poisson structures and homogeneous manifolds of the form $G/H$ with $G$ semisimple. In particular, she found geometric realizations of the Hamiltonian flows as curve flows in $G/H$. The case $G = \rm{PSL}(n+1)$ produces projective realizations of the AGD flow. From the work in \cite{OST} one can see that the continuous limit of the pentagram map itself is the projective realization of the Boussinesq equation. 

In this paper we investigate possible generalizations of the pentagram map to higher dimensional projective spaces. Some discretizations of AGD flows have already appeared in the literature (see for example \cite{LN}), but it is not clear to us how they are related to ours. In particular,  we investigate maps defined as intersections of different types of subspaces in $\RP^m$ whose continuous limit are projective realizations of AGD flows. In section 2 we describe the connection between projective geometry and AGD flows and, in particular, we detail the relation between the AGD Hamiltonian functional and the projective realization of the AGD flow. In section 3 we describe the pentagram map a other maps with the same continuous limit as a simple case in which to describe our approach to finding these discretizations. In section 4 we describe some of the possible generalizations.

In particular, in Theorems 4.1 and 4.2 we show that the projective realization of the AGD flow associated to 
\[
\HH(L) = \int_{S^1}\rm{res}(L^{2/(m+1)}) dx
\]
has a discretization defined through the intersection of one segment and one hyperplane in $\RP^m$. In section 4.2 we analyze the particular cases of $\RP^3$ and $\RP^4$. We show that the projective realization of the flow associated to $\HH(L) = \int_{S^1}\rm{res}(L^{3/4})dx$ can be discretized using the intersection of three planes in $\RP^3$, while the one associated to $\HH(L) = \int_{S^1}\rm{res}(L^{3/5})dx$ can be done using the intersection of one plane and two 3-dimensional subspaces in $\RP^4$ (we also show that it cannot be done with a different choice of subspaces). In view of these results, we conjecture that the projective realization of the the AGD Hamiltonian flow corresponding to (\ref{H}) can be discretized using the intersection one $k-1$ subspace and $k-1$, $m-1$-dimensional subspaces in $\RP^m$.

These results are found by re-formulating the problem of finding the discretizations as solving a system of Dyophantine equations whose solutions determine the choices of vertices needed to define the subspaces. These systems are increasingly difficult to solve as the dimension goes up, hence it is not clear how to solve the general case with this algebraic approach. Furthermore, as surprising as it is that we get a solution at all, the solutions are not unique (nor surprising once we get one solution). As we said before, the pentagram map has extraordinary properties and we will need to search among the solutions in this paper to hopefully find the appropriate map that will inherit at least some of them. One should also check which maps among the possibilities presented here are completely integrable in its own right, that is, as discrete system. It would be very exciting if these two aspects were connected, as we feel this should be the way to learn more before attempting the general case. This is highly non-trivial as not even the invariants of twisted polygons in $\RP^3$ are known. 

This paper is supported by the author NSF grant DMS \#0804541.

\section{Projective geometry and the Adler-Gelfand-Dikii flow} 

\subsection{Projective Group-based moving frames}

Assume we have a curve $\gamma:\R \to \RP^{n}$ and assume the curve has a {\it monodromy}, that is, there exists $M\in \PSL(n+1)$ such that $\gamma(x+T) = M\cdot \gamma(x)$ where $T$ is the period. The (periodic) differential invariants for this curve are well-known and can be described as follows. Let $\l:\R\to\R^{n+1}$ be the unique lift of $\gamma$ with the condition 
\begin{equation}\label{norm}
\det(\l, \l',\dots,\l^{(n)}) = 1.
\end{equation}
 $\l^{(n+1)}$ will be a combination of previous derivatives.  Since the coefficient of $\l^{(n)}$ in that combination is the derivative of $\det(\l, \l',\dots,\l^{(n)})$, and hence zero, there exists periodic functions $k_i$ such that
\begin{equation}\label{Wilc}
\l^{(n+1)} + k_{n-1}\l^{(n-1)}+\dots+k_1\l'+k_0\l=0.
\end{equation}
 The functions $k_i$ are independent generators for the projective differential invariants of $\gamma$ and they are usually called the {\it Wilczynski invariants} (\cite{Wi}).
\begin{definition}(\cite{FO}) A $k$th order left - resp. right - {\it group-based} moving frame is a map
\[
\rho: J^{(k)}(\R, \RP^n) \to \PSL(n+1)
\]
equivariant with respect to the {\it prolonged} action of $\PSL(n+1)$ on the jet space $J^{(k)}(\R, \RP^n)$ (i.e. the action $g\cdot(x,u,u',u'', \dots, u^{(k)}) = (g\cdot u, (g\cdot u)', \dots, (g\cdot u)^{(k)}$) and the left - resp. right - action of $\PSL(n+1)$ on itself.

The matrix $K = \rho^{-1}\rho'$ (resp. $K = \rho'\rho^{-1}$) is called the {\it Maurer-Cartan matrix} associated to $\rho$. For any moving frame, the entries of the Maurer-Cartan matrix generate all differential invariants of the curve (see \cite{Hu}). The equation $\rho' = \rho K$ is called the {\it Serret-Frenet equation} for $\rho$.
\end{definition}

The projective action on $\gamma$ induces an action of $\SL(n+1)$ on the lift $\l$. This action is linear and therefore the matrix $\hat\rho = (\l, \l',\dots,\l^{(n)})\in \SL(n+1)$ is in fact a left moving frame for the curve $\gamma$. We can write equation (\ref{Wilc}) as the system $\hat\rho' = \hat\rho \hat K$
where
\begin{equation}\label{wilcinv}
\hat K = \begin{pmatrix} 0&0&\dots&0&-k_0\\ 1&0&\dots&0&-k_1\\ \vdots&\ddots&\ddots&\vdots&\vdots\\ 0&\dots&1&0&-k_{n-1}\\0&\dots&0&1&0\end{pmatrix},
\end{equation}
is the Maurer-Cartan matrix  generating the Wilczynski invariants.

Group-based moving frames also provide a formula for a general invariant evolution of projective curves. The description is as follows.

We know that $\RP^n \approx \PSL(n+1)/H$, where $H$ is the isotropic subgroup of the origin. For example, if we choose homogeneous coordinates in $\RP^n$ associated to the lift $u \to \begin{pmatrix} 1\\ u\end{pmatrix}$, the isotropic subgroup $H$ is given by matrices $M\in \SL(n+1)$ such that $e_k^T M e_1 = 0$ for $k = 2,\dots,n+1$, and we can choose a section $\varsigma: \RP^n\to \PSL(n+1)$
\[
\varsigma (u) = \begin{pmatrix} 1&0\\ u&I\end{pmatrix}
\]
satisfying $\varsigma(o) = I$. Let $\Phi_g: \RP^n \to \RP^n$ be defined by the action of $g\in \PSL(n+1)$ on the quotient, that is, $\Phi_g(x) = \Phi_g([y])= [gy] = g\cdot x$. The section $\varsigma$ is compatible with the action of $\PSL(n+1)$ on $\RP^n$, that is, 
\begin{equation}\label{actiondef}
g \s(u) = \s(\Phi_g(u)) h
\end{equation}
 for some $h \in H$. If $\h$ is the subalgebra of $H$,  consider the splitting of the Lie algebra
$\g =  \h\oplus \m$, where $\m$ is not, in general, a Lie subalgebra. Since $\s$ is a section, $d\s(o)$ is an isomorphism between $\m$ and $T_o\RP^n$. 

The following theorem was proved in \cite{M2} for a general homogeneous manifold and it describes the most general form of invariant
evolutions in terms of left moving frames.
\begin{theorem} \label{invev}Let $\gamma(t,x) \in \RP^n$ be a flow, solution of an invariant evolution of the form
\[
\gamma_t = F(\gamma,\gamma_x, \gamma_{xx}, \gamma_{xxx}, \dots).
\]
Assume the evolution is invariant under the action of $\PSL(n+1)$, that is, $\PSL(n+1)$ takes solutions to solutions.
Let $\rho(t,x)$ be a family of left moving frames along $\gamma(t,x)$ such that $\rho\cdot o = \gamma$.
Then, there exists an invariant family of tangent vectors $\r(t,x)$, i.e., a family depending on the differential
invariants of $\gamma$ and their derivatives, such that
\begin{equation}\label{projev}
\gamma_ t = d\Phi_\rho(o) \r.
\end{equation}
\end{theorem}

Assume we choose the section $\varsigma$ above. If $\gamma$ is a curve in homogenous coordinates in $\RP^n$, then $\l = W^{-\frac1{n+1}}\begin{pmatrix} 1\\ \gamma\end{pmatrix}$ with $W = \det(\gamma',\dots,\gamma^{(n)})$. In that case, if
$\gamma$ is a solution of (\ref{projev}), with $\rho = (\l, \l',\dots,\l^{(n)})$,    after minor calculations one can directly obtain that $\l$ is a solution of
\[
\l_t = (\l',\dots,\l^{(n)})\r + r_0 \l
\]
where, if $\r = (r_i)$ and
\[
r_0 = -\frac1{n+1} \frac{W_t}W-\sum_{s=1}^n\left(W^{-\frac1{n+1}}\right)^{(s)}r_s.
\]
The coefficient $r_0$ can be written in terms of $\r$ and $\gamma$ once the normalization condition (\ref{norm}) is imposed to the flow $\l$.

Summarizing, the most general form for an invariant evolution of projective curves is given by the projectivization of the lifted evolution
\[
\l_t = r_{n}\l^{(n)}+\dots+r_1\l'+r_0\l
\]
where $r_i$ are functions of the Wilczynski invariants and their derivatives, and where $r_0$ is uniquely determined by the other entries $r_i$ once we enforce (\ref{norm}) on $\l$.

\subsection{AGD Hamiltonian flows and their projective realizations}\label{AGD}

Drinfeld and Sokolov proved in \cite{DS} that the Adler-Gelfand-Dickey (AGD) bracket and its symplectic companion are the reduction of two well-known compatible Poisson brackets defined on the space of loops in $\sl(n+1)^\ast$ (by compatible we mean that their sum is also a Poisson bracket). The author of this paper later proved (\cite{M3}) that the reduction of the main of the two brackets can always be achieved for any homogeneous manifold $G/H$ with $G$ semisimple, resulting on a Poisson bracket defined on the space of differential invariants of the flow. The symplectic companion reduces only in some cases. She called the reduced brackets Geometric Poisson brackets. The AGD bracket is the projective Geometric Poisson bracket (\cite{M3}).

Furthermore, geometric Poisson brackets are closely linked to invariant evolutions of curves, as we explain next. 

\begin{definition}
Given a $G$-invariant evolution of curves in $G/H$, $\gamma_t = F(\gamma, \gamma', \gamma'', \dots)$, there exists an evolution on the differential invariants induced by the flow $\gamma(t,x)$ of the form $\k_t = R(\k, \k', \k'',\dots)$. We say $\k(t,s)$ is the {\it invariantization} of $\gamma(t,x)$. We also say that $\gamma(t,x)$ is a {\it $G/H$-realization} of $\k(t,x)$.
\end{definition}

 As proved in \cite{M3}, any geometric Hamiltonian evolution with respect to a Geometric Poisson bracket can be realized as an invariant evolution of curves in $G/H$. Furthermore, the geometric realization of the Hamiltonian system could be algebraically obtained directly from the moving frame and the Hamiltonian functional. This realization is not unique: a given evolution $\k_t = R(\k, \k', \dots)$ could be realized as $\gamma_t = F(\gamma, \gamma', \dots)$ for more than one choice of manifold $G/H$. For more information, see \cite{M3}.
 
 We next describe this relation in the particular case of $\RP^n$.

\begin{lemma}\label{gauge}

There exists an invariant gauge matrix $g$ (i.e. a matrix in $\SL(n+1)$ such that its entries are differential invariants) such that
\[
g^{-1} g_x +g^{-1} \hat K g^{-1} = K
\]
where $\hat K$ is as in (\ref{wilcinv}) and where
\begin{equation}\label{K}
K = \begin{pmatrix} 0&\ka_{n-1}&\ka_{n-2}&\dots &\ka_1&\ka_0\\ 1&0&0&\dots &0&0\\ 0&1&0&\dots&0&0\\
\vdots&\ddots&\ddots &\ddots&\vdots&\vdots\\ 0&\dots&0&1&0&0\\ 0&\dots&0&0&1&0\end{pmatrix},
\end{equation}
for some choice of invariants $\ka_i$. The invariants $\ka_i$ form a generating and functional independent set of differential invariants.
\end{lemma}

\begin{proof} 
This lemma is a direct consequence of results in \cite{DS}. In particular, of the remark following Proposition 3.1 and its Corollary (page 1989). 

The authors remark how the choice of canonical form for the matrix $K$ is not unique and other canonical forms can be obtained using a gauge. In their paper the matrix $K$ is denoted by $q^{can}+\Lambda$ where $\Lambda = \sum_1^{n-1} e_{i+1,i}$. If we denote by $\bal$ the subalgebra of upper triangular matrices in $\SL(n)$ and by $\bal_r$ the ith diagonal (that is, matrices $(a_{ij})$ such that $a_{ij}=0$ except when $j-i = r$) then we can choose $q^{can} = \sum q_r$ where $q_r \in V_r$ and $V_r$ are $1$-dimensional vector subspaces of $\bal_r$ satisfying $\bal_r = [\Lambda, \bal_{r+1}] \oplus V_r$, $r=0, \dots, n-1$.

Since $[\Lambda, \bal_{r+1}]$ is given by diagonal matrices in $\bal_r$ whose entries add up to zero, there are many such choices, and one of them is the one displayed in (\ref{K}). 
 \end{proof}
Straight calculations show that in the $\RP^3$ case ($n=3$) $k_2 = -\ka_2$, $k_1 = -\ka_1-\ka_2'$ and $k_0=-\ka_0-\ka_1'$, while
\begin{equation}\label{g3}
g = \begin{pmatrix} 1&0&k_2&k_1-k_2'\\ 0&1&0&k_2\\ 0&0&1&0\\0&0&0&1\end{pmatrix}.
\end{equation}
In the $\RP^4$ case $k_3 =-\ka_3$, $k_2 = -\ka_2-3\ka_3'$, $k_1= -\ka_1-2\ka_2'-3\ka_3''$ and $k_0 = -\ka_0-\ka_1'-\ka_2''-\ka_3'''-\ka_3\ka_3'$, while
\begin{equation}\label{g4}
g = \begin{pmatrix}1&0&k_3&k_2-2k_3'&k_1-k_2'-k_3''\\ 0&1&0&k_3&k_2-k_3'\\0&0&1&0&k_3\\0&0&0&1&0\\0&0&0&0&1\end{pmatrix}
\end{equation}.

One can also easily see that the first nonzero upper diagonal of $g$ is always given by $k_{m-1}$ for any dimension $m$.

\begin{theorem}\label{projreal} (\cite{M2}, \cite{M3}) Let $\hat\rho$ be the Wilczynski moving frame, and let $\rho =\hat\rho g$ be the moving frame associated to $K$, where $g$ and $K$ are given as in Lemma \ref{gauge}. Then, the invariant evolution 
\begin{equation}\label{uevpr}
 u_ t = d\Phi_\rho(o) \r 
 \end{equation}
 is the projective realization of the evolution
 \[
 \ka_t = P \r
 \]
 where $P$ is the Hamiltonian operator associated to the AGD bracket as given in \cite{A}. 
 \end{theorem}
The choice of $K$ is determined by the invariants ${\bf\ka}$ being in the dual position to the tangent to the section $\varsigma$ (see \cite{M1}). According to this theorem, in order to determine the projective realizations of AGD Hamiltonian flows, we simply need to fix the moving frame $\rho$ using $\hat\rho$ and $g$, and to find the Hamiltonian functional $\HH$ corresponding to the flow. Notice that, if $\k_t = P\delta_\kappa\HH$ for some Hamiltonian operator $\HH$, then the lift of (\ref{uevpr}) is given by
\begin{equation}\label{liftev}
\l_t = \rho\begin{pmatrix} r_0\\ \r\end{pmatrix} = \hat\rho g\begin{pmatrix}r_0 \\ \delta_\ka\HH\end{pmatrix}
\end{equation}
where $r_0$ is determined uniquely by property (\ref{norm}).

Next we recall the definition of the AGD Hamiltonian functionals.
Let 
\begin{equation}\label{Winv}
L = D^{n+1} + k_{n-1}D^{n-1}+\dots + k_1 D+k_0
\end{equation}
be a scalar differential operator, where $D= \frac d{dx}$, and assume $k_i$ are all periodic. The AGD flow is the AGD Hamiltonian evolution with Hamiltonian functional given by 
\[
\HH(\k) = \int_{S^1} \mathrm{res}(L^{\frac r{n+1}}) dx
\]
$r=2,3,4,\dots n$, where res stands for the residue (or coefficient of $D^{-1}$) of the pseudo-differential operator $L^{r/(n+1)}$. one often refers to the AGD flow as the one associated to $r = n$, but any choice of $r$ will define {\it biHamiltonian} flows (see \cite{DS}). Therefore, any choice of $r$ will result in  a completely integrable flow. For given particular cases this Hamiltonian can be found explicitly. 

\begin{proposition} The AGD Hamiltonian functional for $n=3=r$,  is given by
\[
\HH(\k) = \frac34\int_{S^1} k_0-\frac18 k_2^2 dx
\]
while the one for $n=4=r$ is given by
\[
\HH(\k) = \frac45 \int_{S^1} k_0-\frac{1}5 k_2k_3 dx.
\]
\end{proposition}
\begin{proof} To prove this proposition we simply need to calculate the corresponding residues. We will illustrate the case $n = 3$, and $n=4$ is identically resolved. Assume
\[
L = D^4+k_2 D^2+k_1D+k_0
\]
so that
\[
L^{1/4} = D + \ell_1 D^{-1}+\ell_2D^{-1}+\ell_3D^{-3}  + o(D^{-3}).
\]
Using the relation $L = (L^{1/4})^4$ we can find the uniquely determined coefficients to be $\ell_1 = \frac 14 k_2$, $\ell_2 = \frac14(k_1-\frac32 k_2')$ and $\ell_3 = \frac14(k_0+\frac54k_2''-\frac32 k_1'-\frac38 k_2^2)$. We also find that
\[
L^{3/4} = D^3 + 2\ell_1D+3(\ell_1'+\ell_2)+(\ell_1''+3\ell_2'+3\ell_1^2+3\ell_3)D^{-1} + o(D^{-1}).
\]
Using the periodicity of the invariants, we conclude that
\[
\HH(\k) = \int_{S^1} (\ell_1''+3\ell_2'+3\ell_1^2+3\ell_3) dx =  3 \int_{S^1} (\ell_1^2+\ell_3)dx = \frac34\int_{S^1} k_0-\frac18 k_2^2 dx.
\]
\end{proof}

Notice that, if we are to connect these Hamiltonian flows to their projective realizations, we will need to express their Hamiltonian functionals in terms of the invariants $\ka_i$. In such case $\delta_\ka \HH = \r$ for the Hamiltonian functional defining our system in the new coordinates. If we wish to write the projective realizations in terms of Wilczynski invariants, we can always revert to them once the realizations are found.

\section{The pentagram map} The pentagram map takes its name from an apparently classical result in projective geometry. If we have any given convex pentagon $\{x_1, x_2, x_3, x_4, x_5\}$, we can associate a second pentagon obtained from the first one by joining $x_i$ to $x_{i+2}$ using a segment $\overline{x_{i}x_{i+2}}$,  and defining $x_i^\ast$ to be the intersection of $\overline{x_{k-1}x_{k+1}}$ with $\overline{x_{k}x_{k+2}}$ as in  figure 1. The new pentagon $\{x_1^\ast,x_2^\ast,x_3^\ast,x_4^\ast,x_5^\ast\}$ is projectively equivalent to the first one, that is, there exists a projective transformation taking one to another. The map $T(x_i) = x_i^\ast$ is the pentagram map defined on the space of closed convex $n$-gons, denoted by $\C_n$ . If instead of a pentagon ($n=5$) we consider a hexagon ($n=6$), we need to apply the pentagram map twice to obtain a projectively equivalent hexagon. See \cite{OST} for more details.

\vskip 2ex
\centerline{\includegraphics[height=1.3in]{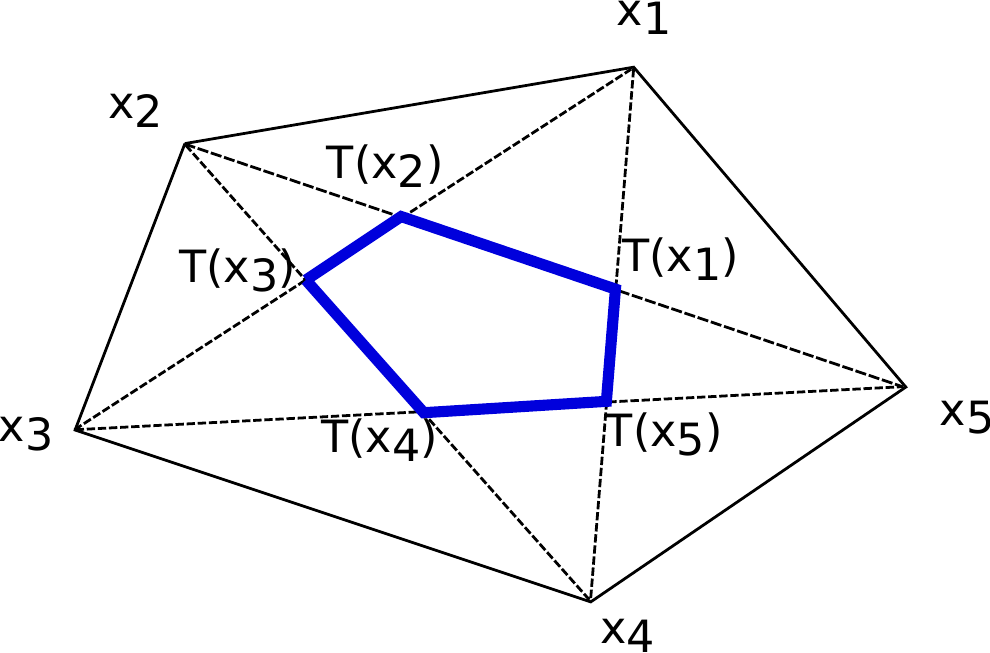}}
\vskip 1ex
\centerline{Fig. 1}
\vskip 2ex

If we consider the polygons to exist in the projective plane instead (all the constructions are transferable to the projective counterpart), we conclude that the pentagram map is the identity on $\C_5$ in $\RP^2$ and involutive on $\C_6$. This property does not hold for any $n$-gon. In fact, in general the map exhibits a quasi-periodic behavior similar to that of completely integrable systems (as shown in \cite{OST}).

The authors of \cite{OST} defined the pentagram map on the space of polygons {\it with a monodromy}. They called these {\it $n$-twisted polygons} and denoted them by $\P_n$. These are infinite polygons such that $x_{i+n} = M(x_i)$ for all $n$ and for some projective automorphism $M$ of $\RP^2$. They proved that, defined on this space and when written in terms of projective invariants of $\P_n$, the pentagram map was completely integrable in the Arnold-Liouville sense. They also proved that, again when written in terms of projective invariants, the pentagram map was a discretization of {\it the projective realization} of the Boussinesq equation, a well-known completely integrable PDE. Furthermore, their calculations show that $T$ is a discretization of the projective realization of the Boussinesq equation, as previously described; in fact, it shows that a certain unique lift of $T$ to $\R^3$ is a discretization of the unique lift of the projective realization given in the previous section.

The pentagram map is not the only one that is a planar discretization of the Boussinesq evolution. In fact, other combinations of segments also are. For example, instead of the pentagram map, consider the following map: $T:\P_m \to \P_m$ where $T(x_n) = \overline{x_{n-2}x_{n+1}}\cap \overline{x_{n-1}x_{n+2}}$. Clearly this map coincides with the pentagram map when defined over pentagons, and it will be degenerate over $\C_{2r}$ for any $r$.

\begin{proposition}\label{syst2} The map $T$ is also a discretization of the projective realization of Boussinesq's equation.
\end{proposition}
\begin{proof} Although the result is intuitive and the proof can perhaps be done in a simpler form as in \cite{OST}, the following process will be carried out in higher dimensions and it is perhaps simpler when illustrated in dimension 2. Instead of working with the projective realization of the Boussinesq equation, we will prove that a unique lifting of the projective realization to $\R^3$ is the continuous limit of the corresponding  unique lifting of the  map $T$ to $\R^3$. Thus, we move from the integrable system to its projective realization and from there to its unique lift to find continuous limits at that level. 

Assume $\Gamma(x)$ is a continuous map on $\R^3$ and assume $\det(\Gamma, \Gamma', \Gamma'') = 1$ so that $\Gamma$ can be considered as the unique lift of a projective curve $\gamma$, as in \cite{OST}.  Let $x_{n+k} = \gamma(x+k\epsilon)$ and assume $\gamma_\epsilon$ is the continuous limit of the map $T$ above (as in \cite{OST} we are discretizing both $t$ and $x$). Denote by $\Gamma_\e$ the unique lift of $\gamma_\e$ to $\R^3$ as in the previous section and assume further than 
\[
\Gamma_\epsilon = \Gamma + \epsilon A+\epsilon^2 B + \epsilon^3 C +o(\epsilon^3)
\]
where $A = \sum_{i=0}^2 \alpha_i \Gamma^{(i)}$, $B = \sum_{i=0}^2 \beta_i \Gamma^{(i)}$ and $C =\sum_{i=0}^2 \gamma_i\Gamma^{(i)}$. Then, the definition of $T$ assumes that $\Gamma_{\epsilon}$ lies in the intersection of both the plane generated by $\Gamma(x-\epsilon)$ and $\Gamma(x+2\epsilon)$ and the one generated by $\Gamma(x-2\epsilon)$ and $\Gamma(x+\epsilon)$. If $\Gamma_\epsilon = a_1 \Gamma(x-\epsilon) + a_2 \Gamma(x+2\epsilon) = b_1 \Gamma(x-2\epsilon)+b_2\Gamma(x+\epsilon)$ for some functions $a_i(x, \epsilon), b_i(x, \epsilon)$, $i=1,2$, then equating the coefficients of $\Gamma, \Gamma'$ and $\Gamma''$ we obtain the relations
\begin{eqnarray*}\label{2d}
1+ \alpha_0\epsilon +  \beta_0\epsilon^2 + o(e^2) &=& a_1+a_2 + o(\epsilon^3) = b_1+b_2+o(\epsilon^3)\\
 \alpha_1\epsilon + \beta_1 \epsilon^2 + o(\epsilon^2) &=& (-a_1+2 a_2)\epsilon + o(\epsilon^3) = (-2b_1+b_2)\epsilon +o(\epsilon^3)\\
\alpha_2\epsilon + \beta_2\epsilon^2 +\gamma_2 \epsilon^3 +o(\epsilon^3) &=& (\frac12 a_1+2a_2)\epsilon^2 +o(\epsilon^3) =  (2b_1+\frac12b_2) \epsilon^2 + o(\epsilon^3).
\end{eqnarray*}
Here we use the fact that $\Gamma'''$ is a combination on $\Gamma$ and $\Gamma'$ according to the normalization $\det(\Gamma, \Gamma', \Gamma'') = 1$ to conclude that the remaining terms are at least $o(\epsilon^3)$. We obtain immediately $\alpha_2 = 0$. 

Let us denote $\a = (a_1, a_2)^T$ and $\b = (b_1, b_2)^T$, and $\a = \sum \a_i \epsilon^i$, $\b = \sum \b_i \epsilon^i$. Then we have the following relations
\[
 \begin{pmatrix}1&1\\-1&2\end{pmatrix} \a_0 = \begin{pmatrix}1&1\\-2&1\end{pmatrix}\b_0 = \begin{pmatrix} 1\\ \alpha_1\end{pmatrix};
\hskip 1ex
 \begin{pmatrix}1&1\\-1&2\end{pmatrix} \a_1 = \begin{pmatrix}1&1\\-2&1\end{pmatrix}\b_1 = \begin{pmatrix} \alpha_0\\ \beta_1\end{pmatrix}
\]
and also the extra conditions
\begin{equation}\label{extra}
\begin{pmatrix}\frac12&2\end{pmatrix}\a_0 = \begin{pmatrix}2&\frac12\end{pmatrix}\b_0 = \beta_2; \hskip 1ex \begin{pmatrix}\frac12&2\end{pmatrix}\a_1 = \begin{pmatrix}2&\frac12\end{pmatrix}\b_1 = \gamma_2.
\end{equation}
Solving for $\a_i$ and $\b_i$ and substituting in the first condition in (\ref{extra}) we get 
\[
\beta_2 = 1 + \frac12 \alpha_1 = 1-\frac12\alpha_1
\]
which implies $\beta_2 = 1$ and $\alpha_1 =0$. 

Finally, substituting in the second condition in (\ref{extra}) we get
\[
\gamma_2 = \alpha_0 + \frac12 \beta_1 = \alpha_0 - \frac12 \beta_1
\]
which results in $\gamma_2 = \alpha_0$ and $\beta_1 = 0$. The final conclusion comes from imposing the lifting condition $\det(\Gamma_\epsilon, \Gamma_{\epsilon}',
\Gamma_\epsilon'') = 1$ to our continuous limit. When expanded in terms of $\epsilon$, and after substituting $A = \alpha_0 \Gamma$, $B = \Gamma'' + \beta_0\Gamma$ we obtain  
\[
1 = 1+3\alpha_0\epsilon + \left(3\alpha_0^2 + 3\beta_0 + \det(\Gamma, \Gamma',\Gamma^{(4)}) +\det(\Gamma, \Gamma''', \Gamma'')\right)\epsilon^2 + o(\epsilon^2).
\]
Using the fact that $\det(\Gamma, \Gamma',\Gamma^{(4)}) +\det(\Gamma, \Gamma'', \Gamma''') = 0$, and the Wilczynski relation $\Gamma''' = -k_1\Gamma' - k_0\Gamma$,
we obtain $\alpha_0 = 0$ and $\beta_0 = -\frac23 k_0$. From here, $\Gamma_e = \Gamma + \epsilon^2 (\Gamma'' -\frac23 k_0 \Gamma)+o(\epsilon^2)$. The result of the theorem is now immediate: it is known (\cite{M2}) that the evolution $ \Gamma_t = \Gamma'' -\frac23 k_0 \Gamma$ is the lifting to $\R^3$ of the projective geometric realization for the Boussinesq equation. (One can also see this limit for the pentagram map in \cite{OST}.)
\end{proof}

\section{Completely Integrable Generalizations of the pentagram map}

\subsection{Discretizations of an $n$-dimensional completely integrable system with a second order projective realization}

In this section we will describe discrete maps defined on $\P_r\subset \RP^{m}$ for which the continuous limit is given by a second order projective realization of a completely integrable PDE. As we described before, we will work with the lift  of the projective flows. Assume $m\ge 2$.

Let $\{x_n\}\in \P_r$. Define $\D_m$ to be a $m-1$ linear subspace determined uniquely by the points $x_n$ and $x_{n+k_i}$ $i=1,\dots m-1$, where $k_i$ are all different from each other and different from $\pm 1$. For example, if $m = 2s-1$, we can choose $x_{n-s}, x_{n-s+1}, \dots, x_{n-2}, x_n, x_{n+2}, \dots, x_{n+s}$. Assume that for every $n$ this subspace intersects the segment $\overline{x_{n-1} x_{n+1}}$ at one point. We denote the intersection $T(x_n)$ and define this way a map $T:\P_r \to \P_r$. The example $k_1 = -2$, $k_2 = 2$ for the case $m = 3$ is shown in figure 2.
\vskip 2ex
\centerline{\includegraphics[height=1.3in]{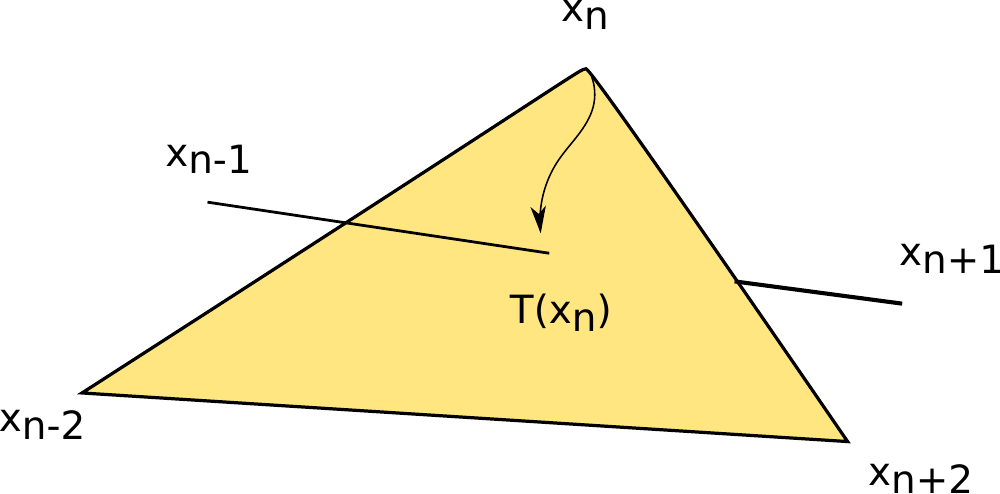}}
\vskip 1ex
\centerline{Fig. 2}
\vskip 2ex
Let $\gamma: \R \to \RP^{m}$ and let $\Gamma: \R \to \R^{m+1}$ be the unique lift of $\gamma$, as usual with the normalization condition 
(\ref{norm}). 
 Following the notation in \cite{OST} we call as before $\gamma(x+k\epsilon) = x_{n+k}$ and denote by $\gamma_\epsilon$ the continuous limit of the map $T$. We denote by $\Gamma_\e$ its lifting to $\R^{m+1}$.

\begin{theorem}\label{basictheorem} The lift to $\R^{m+1}$ of the map $T$ is given by $\Gamma_\epsilon = \Gamma + \frac12\epsilon^2 (\Gamma'' -\frac {2}{m+1} k_{m-1}\Gamma)+o(\epsilon^2)$. 
\end{theorem} 
\begin{proof}
As in the case $n=2$, assume 
\begin{equation}\label{gexp}
\le = \l + \e A + \e^2B+\e^3 C + o(\e^3),
\end{equation}
 and assume further that $A, B$ and $C$ decompose as combinations of $\l^{(k)}$ as 
\[
A = \sum_{r=0}^{m}\al_{r}\l^{(r)}, \hskip 2ex B = \sum_{r=0}^{m}\be_{r}\l^{(r)},\hskip 2ex C = \sum_{r=0}^{m}\mu_{r}\l^{(r)}
\]
Since $T(x_n)$ belongs to the segment $\overline{x_{n-1} x_{n+1}}$, we conclude that the line through the origin representing $\le$ belongs to the plane through the origin generated by $\l(x-\e)$ and $\l(x+\e)$. That is, there exist functions $a$ and $b$ such that 
\[
\le = a \l(x-\e)+b\l(x+\e).
\]
Putting this relation together with the decomposition of $\le$ according to $\e$, we obtain the following $m+1$ equations relating $\al$, $\be$, $a$ and $b$
\begin{eqnarray*}
a+b &=& 1+\al_0\e+\be_0\e^2 +o(\e^2),\\ (a+b)\e^{2r} &=& (2r)! (\al_{2r}\e+\be_{2r}\e^2 +o(\e^2)),\\ (b-a) \e^{2r-1} &=& (2r-1)!(\al_{2r-1}\e+\be_{2r-1}\e^2 +o(\e^2))
\end{eqnarray*}
 $r = 1,\dots, \frac m2$ if $m$ is even, or $r = 0, \dots,\frac{m+1}2$ if $m$ is odd. Directly from these equations we get $\al_r = 0$ for any $r = 2, \dots m$ and $\be_r = 0$ for any $r=3,\dots,m$. We also get $\be_2 = \frac12$, that is, $A = \al_0\l+\al_1\l'$ and $B = \be_0\l+\be_1\l'+\frac12\l''$. The remaining relations involve higher order terms of $C$ and other terms that are not relevant.

Since $T(x_n)$ belongs also to the subspace $\D_m$, $\le$ belongs to the $m$-dim subspace of $\R^{m+1}$ generated by $\l$ and $\l(x+m_i\e)$, $i = 1,\dots,m-1$. That is
\begin{equation}\label{ep1}
\det(\le, \l, \l(x+m_1\e),\dots, \l(x+m_{m-1}\e)) = 0.
\end{equation}
We now expand in $\e$ and select the two lowest powers of $\e$ appearing. These are $\e^{1+\dots +m} = \e^{\frac12m(m+1)}$ and $\e^{\frac12m(m+1)+1}$. They appear as coefficients of $\l^{(r)}$, $r = 0,\dots, m$ situated in the different positions in the determinant. The term involving $\l'$ will come from the $\le$ expansion since this is the term with the lowest power of $\e$.

With all this in mind we obtain that the coefficient of $\e^{\frac12m(m+1)}$ is given by
\[
 X \al_1
\]
for some factor $X$ that we still need to identify. In fact, we only need to know that $X\ne 0$ to conclude that (\ref{ep1}) implies $\al_1 = 0$. 

The factor $X$ corresponds to the coefficient of $\det(\l,\l',\dots,\l^{(m)})$ when we expand
\[
\det\left(\l',\l, \sum_{r=2}^{m}\frac{m_1^{r}}{r!}\l^{(r)},\dots,\sum_{r=2}^{m} \frac{m_{m-1}^{r}}{r!}\l^{(r)}\right).
\]
When one looks at it this way it is clear that $X$ is the determinant of the coefficients in the basis $\{\Gamma, \Gamma',\dots,\Gamma^{(m)}\}$; that is, the determinant of the matrix
\[
\begin{pmatrix} 0&1&0&0&\dots&0\\1&0&0&\dots&0\\ 0&0&\frac{m_1^2}{2!}&\frac{m_1^3}{3!}&\dots&\frac{m_1^{m}}{m!}\\ \dots&\dots&\dots&\dots&\dots&\dots\\ 0&0&\frac{m_{m-1}^2}{2!}&\frac{m_{m-1}^3}{3!}&\dots&\frac{m_{m-1}^{m}}{m!}\end{pmatrix}.
\]
Using expansion and factoring $m_i^2$ from each row, one can readily see that this determinant is a nonzero multiple of the determinant of the matrix 
\[
\begin{pmatrix} 1 & m_1 &m_1^2&\dots &m_1^{m-2}\\ 1&m_2&m_2^2&\dots&m_2^{m-2}\\ \vdots&\vdots&\vdots&\dots&\vdots\\ 1&m_{m-1}&m_{m-1}^2&\dots&m_{m-1}^{m-2}\end{pmatrix}.
\]
This is a Vandermonde matrix with nonzero determinant whenever $m_i\ne m_j$, for all $i\ne j$. Therefore $X\ne 0$ and $\al_1 = 0$. 
Using the normalization condition (\ref{norm}) for $\le$, we obtain
\[
1 = \det(\le,\le',\dots,\le^{(m)}) = 1 + \e(\det(\al_0\l,\l',\dots,\l^{(m)})
\]
\[+\dots+\det(\l,\l',\dots,\al_0\l^{(m)}))+o(\e) = 1+(m+1)\e\al_0+o(\e)
\]
and therefore $\al_0=0$ and $A = 0$.

Finally, the coefficient of $\e^{\frac12m(m+1)+1}$ in (\ref{ep1}) is given by 
\[
X\be_1 + \frac12 Y
\]
where $Y$ is a sum of terms that are multiples of only one determinant, namely $\det(\l,\l',\l'',\dots,\l^{(m-1)},\l^{(m+1)})$ (since $\l''$ carries $\e^2$  in $B$, one of the derivatives in the remaining vectors needs to be one order higher to obtain $\e^{\frac12m(m+1)+1}$. This determines $Y$ uniquely). This determinant is the derivative of $\det(\l,\l',\dots,\l^{(m)}) = 1$ and hence zero. Therefore, since $X\ne 0$ be also obtain $\be_1 = 0$. This result will be true for {\it any different choice of the vertices} when constructing $\D_m$, as far as $x_n$ is belongs to $\D_m$ and our choice is non-singular, that is, as far as $\D_m$ intersects $\overline{x_{n-1} x_{n+1}}$.

Using (\ref{norm}) for $\le$ again, together with $B = \frac12 \l'' + \be_0\l$ and $\l^{(m+1)} = \sum_{r=0}^{m-1} k_r \l^{(r)}$, results in the relation
\[
(m+1) \be_0 + \frac12 2 k_{m-1} = 0
\]
and from here we obtain the continuous limit as stated in the theorem.\end{proof}

One can check that the segment $\overline{x_{n-1}x_{n+1}}$ can be substituted for any choice of the form $\overline{x_{n-r}x_{n+r}}$ to obtain the continuous limit $\le = \l + r!(\Gamma'' -\frac 2{m+1} k_{m-1}\Gamma)\e^2 + o(\e^2)$ instead. Any segment choice of the form $\overline{x_{n+r}x_{n+s}}$, $s\ne -r$, will give us different evolutions, most of which (although perhaps not all) will be non integrable.

\begin{theorem} The invariantization of the projective evolution corresponding to the lifted curve evolution $\Gamma_t = \Gamma'' -\frac 2{m+1} k_{m-1}\Gamma$ is a completely integrable system in the invariants $\k$. 
\end{theorem}
\begin{proof} 

The proof of this theorem is a direct consequence of section \ref{AGD}. Indeed, if 
\[
L= D^{n}+k_{n-2}D^{n-2}+\dots+ k_1D+k_0
\]
and if 
\[
L^{1/{n}} = D + \ell_1D^{-1} + \ell_2D^{-2}  +o(D^{-2})
\]
then, 
\[
L^{2/n} = D^2 + 2\ell_1+(\ell_1'+2\ell_2)D^{-1}+ o(D^{-1})
\]
 and from here the Hamiltonian $\HH(L) = \int_{S^1} \mathrm{res}(L^{2/n}) dx$ is $\HH(L) = 2\int_{S^1} \ell_2 dx$. As before, we can use the relation $(L^{1/n})^n = L$ to find the value for $\ell_2$. Indeed, a short induction shows that
 \[
 L^{k/n} = D^k+k\ell_1D^{k-2}+\left(k\ell_2+\binom{k}2\ell_1'\right)D^{k-3}+o(D^{k-3})
 \]
 which implies $\ell_1 = \frac1n k_{n-2}$ and $\ell_2 = \frac1n k_{n-3}-\binom n2\frac1{n^2} k_{n-2}'$. This implies
 \[
 \HH(L) = \int_{S^1} \mathrm{res}(L^{2/n}) dx =\frac2n \int_{S^1} k_{n-3} dx
 \]
 with an associated variational derivative given by 
 \begin{equation}\label{HL}
 \delta_k \HH = \frac 2ne_{2}.
 \end{equation} 
 It is known (\cite{DS}) that Hamiltonian evolutions corresponding to Hamiltonian functionals $\HH(L) = \int_{S^1} \mathrm{res}(L^{k/n}) dx$, {\it for any $k$}, are biHamiltonian systems and completely integrable in the Liouville sense. Denote by $\delta_{\ka}\HH$ the corresponding variational derivative of $\HH$, with respect to $\ka$.
 
Recall than from (\ref{liftev}), this Hamiltonian evolution has an $\RP^{n-1}$ projective realization  that lifts to the evolution
\[
\l_t = \begin{pmatrix}\l&\l'&\l''&\dots, \l^{(n-1)}\end{pmatrix} g\begin{pmatrix} r_0\\ \delta_\ka \HH\end{pmatrix}
\]
where $g$ is given as in lemma \ref{gauge} and where $r_0$ is uniquely determined by the normalization condition (\ref{norm}) for the flow. 
Recall also that according to the comments we made after lemma \ref{gauge}, the upper diagonal of $g$ is given by the entry $k_{n-2}$. That means the lift can be written as
\[
\l_t = \begin{pmatrix}\l&\l'&\l''+k_{n-2}\l&\dots\end{pmatrix} \begin{pmatrix} r_0\\ \delta_\ka \HH\end{pmatrix}.
\]
The normalization condition  imposed here to find $r_0$ is the exact same condition we imposed on $\le$ to obtain the coefficient of $\l$ in the continuous limit, and they produce the same value of $r_0$. Therefore, we only need to check that $\delta_\ka \HH$ is also a multiple of $e_2$. The change of variables formula tells us that
\begin{equation}\label{changeofvariable}
\delta_\ka\HH = \left(\frac{\delta \k}{\delta \ka}\right)^\ast \delta_k \HH
\end{equation}
and one can easily see from the proof in lemma \ref{gauge} that
\begin{equation}\label{kka}
\frac{\delta \k}{\delta \ka} = \begin{pmatrix} -1&0&\dots &0\\ \ast & -1&\dots&0\\ \vdots& \ddots& \ddots&\vdots\\ \ast&\dots&\ast&-1\end{pmatrix}
\end{equation}
where the diagonal below the main one has entries which are multiple of the differential operator $D$. Clearly, $\left(\frac{\delta \ka}{\delta \ka}\right)^\ast e_2 = e_2$. 

In conclusion, our integrable system has a geometric realization with a lifting of the form $\l_t = \frac1n(\l'' - \frac2n k_{n-2}\l)$. When $n = m+1$ the theorem follows.
\end{proof}
\subsection{Discretizations of higher order Hamiltonian flows in $\RP^3$ and $\RP^4$}

As we previously said, all Hamiltonian evolutions with Hamiltonian functionals of the form (\ref{H}) induce biHamiltonian and integrable systems in the invariants $\k$. When looking for discretizations of these flows, the first thing to have in mind is that these flows have projective realizations of order higher than $2$. That means their lifts - the evolution of $\Gamma$ - will involve $\Gamma^{(r)}$, $r>2$. From the calculations done in the previous section we learned two things: first, the evolution appearing in the coefficient $B$ will be second order, and so any hope to recover a third order continuous limit will force us to seek maps for which $A = B = 0$, and to look for the continuous limit in the term involving $\e^3$. That is, if we go higher in the degree, we will need to go higher in the power of $\e$. Second, no map defined through the intersection of a segment and a hyperplane will have $B = 0$ since that combination forces $B$ to have $\Gamma''$ terms. Thus, we need to search for candidates among intersections of other combinations of subspaces.

\subsubsection{The $\RP^3$ case}

A simple dimensional counting process shows us that in $\R^4$ three $3$-dimensional hyperplanes through the origin generically intersect in one line, so does one $2$-dimensional subspace and a $3$-dimensional hyperplane. These are the only two cases for which the generic intersection of subspaces is a line. The projectivization of the first case shows us that the intersection of three projective planes in $\RP^3$ is a point, and the second case is the one we considered in the previous section. Thus, we are forced to look for discretizations among maps generated by the intersection of three planes in $\RP^3$. Since there are many possible such choices, we will describe initially a general choice of planes and will narrow down to our discretization, which turned out to be not as natural as those in the planar case. The calculations below will show that the choices of these planes need to be very specific to match the evolution associated to the AGD flows. This is somehow not too surprising: integrable systems are rare and their Hamiltonians are given by very particular choices. One will need to tilt the planes in a very precise way to match those choices.

\vskip 2ex
\centerline{\includegraphics[height=1.5in]{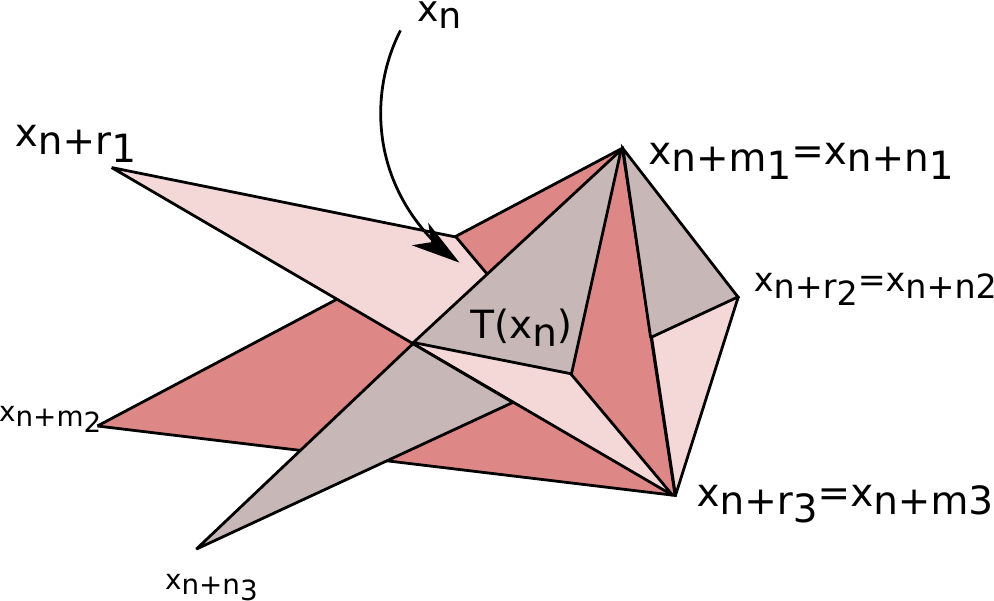}}
\vskip 1ex
\centerline{Fig. 4}
\vskip 2ex

As before, assume $x_{n+k} = \l(x+k\e)$ and consider three projective planes $\p_1$, $\p_2$, $\p_3$ intersecting at one point that we will call $T(x_n)$. Assume our planes go through the following points
\[
\p_1 = \langle x_{n+m_1}, x_{n+m_2},x_{n+m_3}\rangle, \hskip 2ex \p_2 = \langle x_{n+n_1}, x_{n+n_2}, x_{n+n_3}\rangle, 
\]
\[\p_3 = \langle x_{n+r_1}, x_{n+r_2}, x_{n+r_3}\rangle
\]
for some integers $m_i$, $n_i$, $r_i$ that we will need to determine. Figure 4 shows the particular case when $x_{n+m_1} = x_{n+n_1}$, $x_{n+r_2} = x_{n+n_2}$ and $x_{n+m_3} = x_{n+r_3}$.

As before, denote by $\le(x)$ the lifting of $T(x_n)$. Given that $T(x_n)$ is the intersection of the three planes, we obtain the following conditions on $\le$
\begin{align}\label{cond}
\le &=& a_1\l(x+m_1\e) + a_2\l(x+m_2)+a_3\l(x+m_3)\hskip 15ex \\
&=& b_1\l(x+n_1\e) + b_2\l(x+n_2)+b_3\l(x+n_3)\hskip17ex \\
&=& c_1\l(x+r_1\e) + c_2\l(x+r_2)+c_3\l(x+r_3)\hskip 18ex
\end{align}
for functions $a_i, b_i, c_i$ that depend on $\e$. Also as before, assume $\le = \l+\e A+\e^2B+\e^3C+\e^4D+\e^5E+o(\e^5)$ and assume further that
\[
A = \sum_0^3 \al_i\l^{(i)}, \hskip 1ex B = \sum_0^3 \be_i\l^{(i)},\hskip 1ex C = \sum_0^3 \g_i\l^{(i)},\hskip 1ex D = \sum_0^3 \eta_i\l^{(i)},\hskip 1ex E = \sum_0^3 \d_i\l^{(i)}.
\]
\begin{proposition} If $A = B = 0$, then 
\begin{equation}\label{C1}
m_1m_2m_3 = n_1n_2n_3=r_1r_2r_3.
\end{equation}
 In this case $\g_3 = \frac16 m_1m_2m_3$. Under some regularity conditions, (\ref{C1}) implies $A = B = 0$.
\end{proposition}
\begin{proof} Equating the coefficients of $\l^{(i)}$, $i=0,\dots, 3$, condition (\ref{cond}) implies $\al_2 = \al_3=\be_3 = 0$ and the equations
\begin{align}
\label{one} 1+\al_0\e+\be_0\e^2 &= a_1+a_2 + a_3 +o(\e^3)\hskip 2.5in\\
\label{two} \al_1+\be_1\e+\g_1\e^2  &= a_1 m_1+a_2m_2+a_3m_3 + o(e^2)\\
\label{three} 2(\be_2+\g_2\e+\eta_2\e^2)  &= a_1m_1^2+a_2m_2^2+a_3m_3^2- \frac2{4!}k_2(a_1m_1^4+a_2m_2^4+a_3m_3^4)\e^2 + o(\e^2)
\end{align}
with the additional condition
\begin{equation}\label{condition}
6(\g_3+\eta_3\e+\d_3\e^2)  = a_1m_1^3+a_2m_2^3+a_3m_3^3-\frac1{20} k_2(a_1m_1^5+a_2m_2^5+a_3m_3^5) \e^2 + o(\e^2). 
\end{equation}
The terms with $k_2$ appear when we use the Wilczynski relation $\l^{(4)} = -k_2\l''-k_1\l'-k_0\l$. We obtain similar equations with $b_i, n_i$ and with $c_i, r_i$.

Denote by $\a = (a_i)$ and decompose $\a$ as $\a = \sum \a_i\e^i$, with analogous decompositions for $b_i$ and $c_i$. Then, the  first three equations above allow us to solve for $\a_i$, $\b_i$ and $\c_i$, $i=0,1$, namely
$
\a_i = A(m)^{-1} v_i, \b_i = A(n)^{-1}v_i, \c_i = A(r)^{-1}v_i$ with
\begin{equation}\label{As}
A(s) = \begin{pmatrix}1&1&1\\ s_1&s_2&s_3\\ s_1^2&s_2^2&s_3^2\end{pmatrix}, v_0 = \begin{pmatrix}1\\\al_1\\2\be_2\end{pmatrix}, v_1 = \begin{pmatrix} \al_0\\\be_1\\2\g_2\end{pmatrix}
\end{equation}
We can now use (\ref{condition}) to obtain conditions on the parameters $\al_i$ and $\be_i$. Indeed, these are
\[
6\g_3 = \begin{pmatrix}m_1^3&m_2^3&m_3^3\end{pmatrix} \a_0 = \begin{pmatrix}n_1^3&n_2^3&n_3^3\end{pmatrix} \b_0= \begin{pmatrix}r_1^3&r_2^3&r_3^3\end{pmatrix} \c_0.
\]
After substituting the values for $\a_0, \b_0$ and $\c_0$, these three equations can be alternatively written as the system
\begin{equation}\label{M}
6\g_3 \begin{pmatrix} 1\\1\\1\end{pmatrix}= \begin{pmatrix} M_1&M_2&M_3\\ N_1&N_2&N_3\\ R_1&R_2&R_3\end{pmatrix}\begin{pmatrix}1\\\al_1\\\be_2\end{pmatrix}
\end{equation}
where
\[
M_1 = \frac1{\det(A(m))}\det\begin{pmatrix}m_1^3&m_2^3&m_3^3\\ m_1&m_2&m_3\\ m_1^2&m_2^2&m_3^2\end{pmatrix}, M_2 = \frac1{\det(A(m))}\det\begin{pmatrix}1&1&1\\ m_1^3&m_2^3&m_3^3\\ m_1^2&m_2^2&m_3^2\end{pmatrix},
\]
\[
 M_3 = \frac1{\det(A(m))}\det\begin{pmatrix}1&1&1\\ m_1&m_2&m_3\\ m_1^3&m_2^3&m_3^3\end{pmatrix}
\]
with similar definitions for $N_i$ (using $n_i$ instead of $m_i$) and $R_i$ (using $r_i$). 
\vskip 1ex

Although the following relation must be known, I was unable to find a reference. 
\begin{lemma}\label{Vandermonde} Assume $A(m)$ is an $s\times s$ Vandermonde matrix with constants $m_1, \dots, m_s$. Let $M_i = \frac{\det A_i(m)}{\det A(m)}$, where $A_i(m)$ is obtained from $A(m)$ when substituting the ith row with $m_1^s, \dots, m_s^s$. Then $M_i = (-1)^s p_{s-i+1}$, where $p_r$ are the elementary symmetric polynomials defined by the relation
\[
(x-m_1)\dots (x-m_s) = x^s + p_{s-1} x^{s-1}+\dots + p_1 x+p_0.
\]
\end{lemma}
\begin{proof}[Proof of the lemma] One can see that
\[
(x-m_1)\dots (x-m_s) = (-1)^s \det A(m)^{-1} \det\begin{pmatrix} 1&1&\dots&1\\ x&m_1&\dots&m_s\\ x^2&m_1^2&\dots& m_s^2\\ \vdots&\vdots&\dots&\vdots\\ x^s&m_1^s&\dots&m_s^s\end{pmatrix}
\]
since both polynomials have the same roots and the same leading coefficient. The lemma follows from this relation.
\end{proof}
{\it We continue the proof of the Proposition}. 
From the lemma we know the values of $M_i$ to be 
\begin{equation}\label{Mi}
M_1 = m_1m_2m_3,\hskip 2ex M_2 = -(m_1m_2+m_1m_3+m_2m_3),\hskip 2ex M_3 = m_1+m_2+m_3.
\end{equation}
 It is now clear that, if $A = B = 0$, then $\al_1 = \be_2 = 0$ and the system implies $M_1 = N_1=R_1$ as stated in the proposition.

Notice that the condition $M_1=N_1=R_1$ does not guarantee $A = B = 0$; this will depend on the rank of the matrix in (\ref{M}). 
Notice also that the same condition applies to $\a_1$, $\b_1$ and $\c_1$, we only need to substitute $F_0$ by $F_1$ in the calculations. Therefore, if the rank of 
\begin{equation}\label{check}
\begin{pmatrix} M_2&M_3\\ N_2&N_3\\ R_2&R_3\end{pmatrix}
\end{equation}
is maximal, then $M_1 = N_1= R_1$ if, and only  if $\al_1 = \be_2 = 0$, and also $\be_1 = \g_2 = 0$. Assume that the rank of this matrix is 2. Then, $A = \al_0\l$ and  condition (\ref{norm}) applied to $\le$ as before becomes
\[
0 = \e 4\al_0+o(\e)
\]
implying $\al_0 = 0$ and $A = 0$. Likewise, if the rank is two then $\be_1=\be_2 = 0$ and $B = \be_0\g$. Applying (\ref{norm}) again we will obtain $\be_0=0$ and $B=0$.
\end{proof}

If we now go back to (\ref{one})-(\ref{two})-(\ref{three})-(\ref{condition}) and we compare the powers of $\e^2$, we can solve for $\a_2, \b_2, \c_2$ as $\a_2 = A(m)^{-1}v^a_2$, $\b_2 = A(n)^{-1}v^b_2$, $\c_2 = A(r)^{-1}v^c_2$ with
\[
v^a_2 = \begin{pmatrix} \be_0\\ \g_1\\ 2\eta_2+\frac 1{12}k_2 \m_4\cdot\a_0\end{pmatrix}, v^b_2 = \begin{pmatrix} \be_0\\ \g_1\\ 2\eta_2+\frac 1{12}k_2 \nb_4\cdot\b_0\end{pmatrix}
, v^c_2 = \begin{pmatrix} \be_0\\ \g_1\\ 2\eta_2+\frac 1{12}k_2 \r_4\cdot\c_0\end{pmatrix}
\]
where $\m_4 = (m_1^4, m_2^4, m_3^4)$. From now on we will denote $\m_i = (m_1^i, m_2^i, m_3^i)$ and we will have analogous notation for $\nb_i$ and $\r_i$. Using these formulas in the extra condition (\ref{condition}), we obtain a system of three equations.
As before, we can rearrange these equations to look like the system 
\begin{equation}\label{Meq}
\M\begin{pmatrix} \gamma_1\\ 2\eta_2\\ -3!\delta_3\end{pmatrix}= \begin{pmatrix} \displaystyle \frac1{20}\m_5A^{-1}(m)e_1-\frac{M_3}{12}\m_4A^{-1}(m)e_1\\\\\displaystyle  \frac1{20}\nb_5A^{-1}(n)e_1-\frac{N_3}{12}\nb_4A^{-1}(n)e_1\\\\\displaystyle
 \frac1{20}\r_5A^{-1}(r)e_1-\frac{R_3}{12}\r_4A^{-1}(r)e_1\end{pmatrix} k_2
 \end{equation}
 where
 \[
 \M = \begin{pmatrix} M_2&M_3&1\\ N_2&N_3&1\\ R_2&R_3&1\end{pmatrix}
 \]
 
Since $\gamma_3 =\frac{m_1m_2m_3}{6}$, at first look it seems as if the numbers $(m_1,m_2,m_3) = (-1,-2,3)$, $(n_1,n_2,n_3)=(-1,2,-3)$ and $(r_1,r_2,r_3)=(1,-2,-3)$ would be good choices for a generalization of the pentagram map (it would actually be a direct generalization of the map in proposition \ref{syst2}). As we see next, the situation is more complicated.  

\begin{lemma}\label{square} Assume $m_1^2+m_2^2+m_3^2=n_1^2+n_2^2+n_3^2 = r_1^2+r_2^2+r_3^2$. Then the continuous limit of $\le$ is not the AGD flow and it is not biHamiltonian.
\end{lemma}

\begin{proof} We already have $A = B = 0$ and $\g_2 = 0$. 
Equation (\ref{Meq}) allows us to solve for $\gamma_1$. After that, $\gamma_0$ will be determined, as usual, by the normalization equation (\ref{norm}). Direct calculations determine $\m_5A^{-1}(m) e_1 = M_1M_5$ and $\m_4A^{-1}(m)e_1 = M_1M_3$, where $M_1, M_2, M_3$ are as in (\ref{Mi}) and where
\begin{equation}\label{M5}
M_5 = m_1^2+m_2^2+m_3^2+m_1m_2+m_1m_3+m_2m_3 = M_3^2 + M_2.
\end{equation}
From here we get 
\[
\frac1{20}M_5-\frac1{12}M_3^2 = -\frac1{30}(M_5+M_2) + \frac7{60} M_2.
\]
$M_5+M_2 = m_1^2+m_2^2+m_3^2$, and so $M_5+M_2= N_5+N_2=R_5+R_2$ by the hypothesis of the lemma. We have
\[\gamma_1 = M_1\frac{k_2}{\det\M}\left(-\frac1{30}\det\begin{pmatrix}M_5+M_2&M_3&1\\ N_5+N_2&N_3&1\\R_5+R_2&R_3&1\end{pmatrix} +\frac 7{60}\det\begin{pmatrix}M_2&M_3&1\\N_2&N_3&1\\R_2&R_3&1\end{pmatrix} \right) 
\]
\[= \frac{7M_1}{60}k_2.
\]
Since $\gamma_3 = \frac{M_1}{6}$ we finally have that $\le$ is expanded as
\[
\le = \l + \frac{M_1}6(\l'''+\frac7{10}k_2\l'+r_0\l)\e^3+ o(\e^3)
\]
where $r_0$ is uniquely determined by (\ref{norm}). 

Let's now find the lifting for the projective realization of the AGD flow. From lemma \ref{gauge} the moving frame associated to $K$ is given by $\rho=\hat\rho g$ where $g$ is the gauge matrix in the case of $\SL(4)$, given as in (\ref{g4}). We also know that the lifted action of $\SL(4)$ on $\R^4$ is linear. With this information one can conclude that the lift of the evolution (\ref{uevpr}) to $\R^4$ is of the form
\[
\l_t = \rho \begin{pmatrix}r_0\\ \delta_\ka\HH\end{pmatrix} =  \begin{pmatrix}\l & \l'& \l''+k_2\l & \l'''+k_2\l'+(k_1-k_2')\l\end{pmatrix} \begin{pmatrix}r_0\\ \delta_\ka\HH\end{pmatrix}.
\]
According to (\ref{changeofvariable}), we also know that 
\[
\delta_\ka\HH = \frac{\delta \k}{\delta \ka}^\ast\delta_k\HH = \frac{\delta \k}{\delta \ka}^\ast \begin{pmatrix} -\frac14 k_2\\ 0\\ 1\end{pmatrix} = \begin{pmatrix} \frac14 k_2\\ 0\\ -1\end{pmatrix}.
\]
 Therefore, the lifting of the projective realization of the AGD flow is  $\l_t = -\l''' -\frac34 k_2\l'-r_0\l$, with $r_0$ determined by (\ref{norm}). Although a change in the coefficient of $\l'$ seems like a minor difference, the change is in the Hamiltonian of the evolution and any small change usually results in a system that is no longer biHamiltonian, as it is the case here. The calculations that show that the resulting Hamiltonian is no longer biHamiltonian are very long and tedious, and we are not including them here. 
\end{proof}
 
According to this lemma, in order to find a discretization of the AGD flow we need to look for planes for which the hypothesis of the previous lemma does not hold true.
\begin{theorem} Assume $(m_1,m_2,m_3) = (-c,a,b)$, $(n_1,n_2,n_3)=(c,-a,b)$ and $(r_1,r_2,r_3)=(c,-1,ab)$. Then, the map T is a discretization of the AGD flow whenever 
\begin{equation}\label{level}
\frac{c-1+a(b-1)}{b-c} = -\frac{5}4.
\end{equation} 
\end{theorem}
\begin{proof} With the choices above one can check that
\[
\det\begin{pmatrix} M_2&M_3&1\\ N_2&N_3&1\\ R_2&R_3&1\end{pmatrix} = -2(c-a)(a-1)(b+1)(b-c)
\]
while
\[
\det\begin{pmatrix} M_5&M_3&1\\ N_5&N_3&1\\ R_5&R_3&1\end{pmatrix} = -2(c-a)(a-1)(b+1)(c-1+a(b-1)).
\]
With this particular ansatz, the rank of (\ref{check}) is maximal whenever $c\ne a$, $a\ne 1$, $b\ne -1$ and $b\ne c$.
Using the fact that 
\[
\frac1{20}M_5-\frac1{12}M_3^2 = -\frac1{30}M_5 + \frac1{12} M_2
\]
we obtain
\[
\gamma_1 = -\frac1{30}M_1\frac{c-1+a(b-1)}{b-c} + \frac1{12}M_1.
\]
As we saw before, the lifting of the realization of the AGD flow is given by $\l_t= \l'''+\frac34k_2\l'+r_0\l$ (after a change in the sign of $t$), where $r_0$ is determined by (\ref{norm}). Since $\gamma_3 = \frac{M_1}{6}$, to match this flow we need $\gamma_1 =\frac{1}{8}k_2M_1$.
That is, we need
\[
\frac{c-1+a(b-1)}{b-c} = -\frac{5}4
\]
as stated in the theorem.
\end{proof}

Equation (\ref{level}) can be rewritten as
\[
z + xy = 1
\]
where $z = c$, $x={4a+5}$, $y = 1-b$. In principle there are many choices like $c = -2, a = -2, b=2$ that solve these equations, but these are not valid choices since the planes associated to $(-c,a,b)$, $(c,a,-b)$ are not well defined (they are determined by only two points). Thus, looking for appropriate values is simple, but we have to be careful. In particular, we cannot choose any vanishing value, since $m_1m_2m_3=n_1n_2n_3=r_1r_2r_3$ would imply that all three planes intersect at $x_n$ (the zero value for $m_i$, $n_i$ and $r_i$) and $T$ would be the identity. We also do not want to have the condition in lemma \ref{square} (hence $a$ and $b$ cannot be $\pm1$), plus we want to have the matrix (\ref{check}) to have full rank, which implies $c\ne a$, $a\ne 1$, $b\ne -1$ and $b\ne c$. A possible choice of lowest order is
\[\begin{array}{cc}
a = -2, b=3, c=-5 
\end{array}
\]
and other combinations involving higher values.  Choosing this simplest value, we see that the planes are $\Pi_1 = \langle x_{n-2}, x_{n+3}, x_{n+5}\rangle$, $\Pi_2 = \langle x_{n-5}, x_{n+2}, x_{n+3}\rangle$ and $\Pi_3=\langle x_{n-5}, x_{n+1}, x_{n-6}\rangle$, which shows just how precise one needs to be when choosing them. Of course, these are not necessarily the simplest choices, just the ones given by our ansatz. Using a simple C-program and maple one can show that $6$ is the lowest integer value that needs to be included, so our choice is in fact minimal in that sense (there is no solution if we only use $-5,-4 \dots, 4, 5$ for $m_i$, $n_i$ and $r_i$). One can check that, for example, $\Pi_1 = \langle x_{n+2}, x_{n-3}, x_{n+5}\rangle$, $\Pi_2 = \langle x_{n+5}, x_{n-2}, x_{n+3}\rangle$,  $\Pi_3=\langle x_{n+5}, x_{n+1}, x_{n+6}\rangle$ and  $\Pi_1 = \langle x_{n+1}, x_{n-3}, x_{n-4}\rangle$, $\Pi_2 = \langle x_{n-1}, x_{n-3}, x_{n+4}\rangle$ and $\Pi_3=\langle x_{n+1}, x_{n+2}, x_{n+6}\rangle$ are also choices.

It will be valuable and very interesting to learn the geometric significance (if any) of this condition and whether or not the map $T$, when written in terms of the projective invariants of elements of $\P_n$, is also completely integrable as it is the case with the pentagram map. Learning about their possible discrete structure might aid the understanding of condition (\ref{level}) and would greatly aid the understanding of the general case. Doing this study is non-trivial as even the projective invariants of twisted polygons in three dimensions are not known.

\subsubsection{The $\RP^4$ case} One thing we learned form the $\RP^3$ case is that choosing an appropriate set of linear subspaces intersecting to match discretizations of biHamiltonian flows involve solving Dyophantine  problems. This will also become clear next. These  Dyophantines problems grow increasingly complicated very fast, but, nevertheless they seem to be solvable. Our calculations are telling: a general Dyophantine problem of high order is unlikely to have solutions. Ours can be simplified using the symmetry in the values $m_i$, $n_i$ and $r_i$ to reduce the order as we write them in terms of elementary symmetric polynomials. Still, as we see next, in cases when, based on degree and number of variables, one should expect a solution, we have none. In other apparently similar cases we have an infinite number. Hence, the fact that our cases have solutions hints to a probable underlying reason to why they do.
In this section we find a discretization for a the second integrable AGD flow in $\RP^4$, and we will draw from it a conjecture for the general case.

\begin{proposition} The projective geometric realization of the AGD Hamiltonian system associated to the Hamiltonian
\[
\HH(L) = \int_{S^1} \mathrm{res}(L^{3/5}) dx
\]
has a lift given by
\begin{equation}\label{ev4}
\l_t = \l'''+\frac{27}5k_3\l'+r_0\l
\end{equation}
where again $r_0$ is determined by the property (\ref{norm}) of the flow.
\end{proposition}
\begin{proof} 
As before, we need to find $\delta_k \HH$, change the variable to $\delta_{\ka}\HH$ (to relate it to the coefficients of the realizing flow), and write these coefficients in terms of the Wilczynski invariants $k_i$.

If $L^{1/5} = D+\ell_1D^{-1}+\ell_2D^{-2}+\ell_3D^{-3}+o(D^{-3})$, then
\[
L^{3/5} = D^3+3\ell_1D+3(\ell_1'+\ell_2)+(\ell_1''+3\ell_2'+3\ell_1^2+3\ell_3)D^{-1} + o(D^{-1})
\]
so that
\[
\HH(L) = \int_{S^1}\mathrm{res}(L^{3/5})dx = 3\int_{S^1} (\ell_1^2+\ell_3)dx.
\]
Using $(L^{1/5})^5 = L$ we find directly that
\[
\ell_1 = \frac15 k_3, \hskip 2ex\ell_3 = \frac15(k_1-2k_2'+2k_3''+2k_3^2).
\]
Therefore
\[
\HH(L) = \frac35\int_{S^1} \frac{11}5k_3^2+k_1 dx
\]
and $\delta_k\HH = \frac35(e_3 + \frac {22}5 k_3 e_1)$.

Using expression (\ref{kka}) and Lemma \ref{gauge} we get

\[
\frac{\delta k}{\delta\ka} = \begin{pmatrix} -1&0&0&0\\ -3D&-1&0&0\\ -3D^2&-2D&-1&0\\ -D^3-\ka_3'-\ka_3D&-D^2&-D&-1\end{pmatrix}.
\]
\[
\delta_\ka\HH= \left(\frac{\delta\k}{\delta \ka}\right)^\ast\delta\HH = \begin{pmatrix}-\frac {22}5 k_3\\0\\-1\\0\end{pmatrix}.
\]
Finally, the matrix $g$ appearing in Lemma \ref{gauge} is in this case given by (\ref{g4}) and the lifting of the projective realization associated to the $\HH$ Hamiltonian evolution is given by
\[
\l_t = (\l, \l', \l'', \l''', \l^{(4)}) g \begin{pmatrix}\hat r_0\\\delta_\ka\HH\end{pmatrix} = -\l'''-\frac{27}5k_3\l'-r_0\l.
\]
A simple change of sign in $t$ will prove the theorem.

\end{proof}

There are 4 possible combinations of linear submanifolds in $\RP^4$ intersecting at a point: four $3$-dim subspaces, two $2$-dim planes,  one $2$-dim plane and two $3$-dim subspaces and one line and one $3$-dim subspace. These correspond to the following subspaces through the origin in $\R^5$ that generically intersect in a line
\begin{enumerate}
\item\label{1} Four 4-dimensional subspaces.
\item\label{2} Two 3-dimensional subspaces.
\item\label{3} One 3-dimensional subspace and two 4-dimensional ones.
\item\label{4} One 2-dimensional subspace and one 4-dimensional one.
\end{enumerate}
Case (\ref{4}) corresponds to the intersection of a projective line and hyperplane, the case we studied first. Case (\ref{1}) is a natural choice for the fourth order AGD flow, and one can easily check that it cannot have a third order limit but a fourth order one. Therefore, we have choices (\ref{2}) and (\ref{3}) left.

\begin{theorem} \label{case2} If a nondegenerate map $T:\P^n\to \P^n$ is defined using combination (\ref{2}) in $\RP^4$ and its lifting has a continuous limit of the form $\l_t = a\l'''+ b\l'+c\l$, then $b = \frac3{10}a$ and $c$ is determined by (\ref{norm}).
\end{theorem}

\begin{proof} Two $3$-dim subspaces through the origin in $\R^5$ are determined by three points each. Assume $x_{n+m_i}$, $i=1,2,3$ are the points in one of them, while $x_{n+n_i}$ are the points in the other one. If $\l(x+\e r)= x(n+r)$ for any $r$, then $\l_\e$ belongs to the intersection of these two subspaces whenever
\begin{eqnarray}\label{first}\l_\e &=& a_1\l(x+m_1\e)+a_2\l(x+m_2\e)+a_3\l(x+m_3\e)\\ &=&  b_1\l(x+n_1\e)+b_2\l(x+n_2\e)+b_3\l(x+n_3\e)
\end{eqnarray}
 If, as before, $\l_\e = \l+\e A+\e^2B+\e^3C +\e^4D+ \e^5E+\e^6 F +o(\e^6)$ and $A = \sum_{i=0}^4 \alpha_i\l^{(i)}$, $B = \sum_{i=0}^4 \beta_i\l^{(i)}$, $C = \sum_{i=0}^4 \gamma_i\l^{(i)}$, $D=\sum_{i=0}^4\d_i\l^{(i)}$, $E=\sum_{i=0}^4\eta_i\l^{(i)}$ and $F=\sum_{i=0}^4\nu_i\l^{(i)}$ then $\alpha_2=\al_3=\al_4=0=\be_3=\be_4=\g_4$ and (\ref{first}) can be split into
\begin{eqnarray*}\label{1-2-3}
1+\alpha_0\e+\beta_0\e^2+\g_0\e^3 +o(\e^3) &=& a_1+a_2+a_3+o(\e^3) = b_1+b_2+b_3+o(\e^3)\\
\alpha_1+\e\beta_1+\e^2\g_1+o(\e^2) &=& \m_1\cdot \a+o(\e^2) = \nb_1\cdot \b+o(\e^2)\\
\be_2+\e\g_2+\e^2\d_2+o(\e^2) &=&\frac1{2!}\m_{2}\cdot \a+o(\e^2) = \frac1{2!}\nb_{2}\cdot \b + o(\e^2)
\end{eqnarray*}
corresponding to the coefficients of $\l$, $\l'$ and $\l''$ and
\begin{eqnarray*}\label{4-5}
 \g_3+\e\d_3+\e^2\eta_3 + o(\e^2) &=& \frac1{3!} \m_{3}\cdot \a-\frac{\e^2}{5!}k_3 \m_{5}\cdot \a +o(\e^2)\\ &=& \frac1{3!} \nb_{3}\cdot \b-\frac{\e^2}{5!}k_3 \nb_{5}\cdot \b +o(\e^2)\\ 
 \d_4+\e\eta_4+\e^2\nu_4+o(\e^2) &=& \frac1{4!} \m_{4}\cdot \a-\frac{\e^2}{6!}k_3 \m_{6}\cdot \a + o(\e^2)\\ &=& \frac1{4!} \nb_{4}\cdot \b-\frac{\e^2}{6!}k_3 \nb_{6}\cdot \b + o(\e^2)
 \end{eqnarray*}
corresponding to the coefficients of $\l^{(3)}$ and $\l^{(4)}$. Here we have used the relation $\l^{(5)} = -k_3\l'''-k_2\l''-k_1\l'-k_0\l$ and we have used the notation $\m_r= (m_1^r,m_2^r, m_3^r)$, as we did in the previous case. Likewise with $\n$.

If, as before, we use the notation $\a = \sum_{j=0}^\infty \a_j\e^j$, $\b = \sum_{j=0}^\infty \b_j\e^j$, then the first three equations completely determine $\a_i$ and $\b_i$, $i=0,1,2$. Indeed, they are given by $\a_i = A(m)^{-1}v_i$, $\b_i = A(n)^{-1}v_i$, where $A(s)$ is given as in (\ref{As}) and where
\begin{equation}\label{vi}
v_0 = \begin{pmatrix} 1\\ \al_1\\ 2!\be_2\end{pmatrix}, v_1 = \begin{pmatrix} \al_0\\ \be_1\\ 2!\g_2\end{pmatrix}, v_2 = \begin{pmatrix} \be_0\\ \g_1\\ 2!\d_2\end{pmatrix}.
\end{equation}
 The last four equations above are extra conditions that we need to impose on these coefficients. They can be rewritten as
 \begin{eqnarray}
\label{uno} 3!\g_3 &=& \m_{3}\cdot \a_0 = \nb_{3}\cdot \b_0\\
\label{dos} 3! \d_3 &=&\m_{3}\cdot \a_1 = \nb_{3}\cdot \b_1\\
\label{tres} 3!\eta_3&=& \m_{3}\cdot\a_2-\frac{3!}{5!} k_3\m_{5}\cdot \a_0 = \nb_{3}\cdot\b_2-\frac{3!}{5!} k_3\nb_{5}\cdot \b_0
 \end{eqnarray}
 and
 \begin{eqnarray}
\label{cuatro} 4!\d_4 &=& \m_{4}\cdot\a_0 = \nb_4\cdot \b_0\\ \label{cinco} 4!\eta_4 &=&\m_4\cdot  \a_1=\nb_4\cdot \b_1\\\label{seis} 4!\nu_4&=&\m_4\cdot \a_2-\frac{4!}{6!} k_3 \m_6\cdot\a_0 = \nb_4\cdot \b_2-\frac{4!}{6!} k_3 \nb_6\cdot\b_0.
 \end{eqnarray}
 
 If as before we denote $(m_1^3, m_2^3, m_3^3) A(m)^{-1} = (M_1, M_2, M_3)$, where $M_i$ are the negative of the basic symmetric polynomials as shown in Lemma \ref{Vandermonde}, then equations (\ref{uno}) imply
 \[
 \begin{pmatrix} M_1&M_2&M_3\end{pmatrix}\begin{pmatrix} 1\\\al_1\\2!\be_2\end{pmatrix} =  \begin{pmatrix} N_1&N_2&N_3\end{pmatrix}\begin{pmatrix} 1\\\al_1\\2!\be_2\end{pmatrix} =3!\g_3.
 \]
 Since our continuous limit needs to be third order, and hence appearing in $C$, to have the proper continuous limit we need $A = B = 0$, thus  $\al_1 = \be_2 = 0$ and $M_1 = N_1 = 3!\g_3$.
 
 From the proof of Lemma \ref{square} and direct calculations, we know that $\m_4A(m)^{-1} = (M_1M_3, M_4, M_5)$, where $M_4 = -(m_1+m_2)(m_2+m_3)(m_1+m_3) = M_3M_2+M_1$ and $M_5 = M_3^2+M_2$ is as in (\ref{M5}). Therefore, equation (\ref{cuatro}) can be rewritten as
 \[
  \begin{pmatrix} M_1M_3&M_4&M_5\end{pmatrix}\begin{pmatrix} 1\\\al_1\\2!\be_2\end{pmatrix} =  \begin{pmatrix} N_1N_3&N_4&N_5\end{pmatrix}\begin{pmatrix} 1\\\al_1\\2!\be_2\end{pmatrix} =4!\d_4
  \]
from which we can conclude that, if $\alpha_1 = 0 = \beta_2$, then $M_1M_3 = N_1 N_3$, that is $M_3 = N_3$.  (Notice that $M_1=N_1\ne 0$ since $M_1=N_1=0$ implies both planes intersect at $x_n$ and $T$ is the identity.)

In order to have (\ref{ev4}) as continuous limit we will need $v_1 = 0$ which implies $\a_1 = \b_1 = 0$. Using (\ref{dos}) and (\ref{cinco}) we get $\delta_3 = \eta_4 = 0$. 

Finally $v_2 = \g_1e_2+2!\d_2e_3$ and $\m_5A^{-1}(m)e_1 = M_1 M_5 = M_1(M_3^2+M_2)$ as in the proof of Lemma \ref{square}. Therefore (\ref{tres}) becomes 
\[
3!\eta_3 = M_2\g_1+2!M_3\d_2 - \frac{3!}{5!} M_1(M_3^2+M_2) = N_2\g_1+2!N_3\d_2 - \frac{3!}{5!} N_1(N_3^2+N_2).
\]
Since we already know that $M_1 = N_1$ and $M_3 = N_3$, this equation becomes
\[
M_2(\g_1-\frac{3!}{5!}M_1) = N_2 (\g_1-\frac{3!}{5!}M_1).
\]
From here, either $\g_1= \frac{3!}{5!}M_1 = \frac{(3!)^2}{5!}\g_3$, resulting in the continuous limit displayed in the statement of the Theorem, or $M_2 = N_2$. But, from Lemma \ref{Vandermonde} conditions $M_1 = N_1$, $M_2=N_2$ and $M_3=N_3$ imply $m_1 = n_1, m_2 = n_2$ and $m_3 = n_3$. This is a degenerate case, the three planes are equal.
  \end{proof}
  
\begin{theorem} There exists a map $T: \P^n \to \P^n$  in $\RP^4$ defined using option (\ref{3}) whose continuous limit (as previously defined) is integrable and has a lifting given by (\ref{ev4}). The map is not unique.
\end{theorem}

Before we prove this theorem, let me point at the apparent pattern we see here: The $L^{\frac2{m+1}}$-Hamiltonian flow, $m\ge 2$, is the continuous limit of a map obtained when intersecting a $1$-dim line and a $m-1$-dim subspace of $\RP^m$. The $L^{\frac3{m+1}}$-Hamiltonian flow, $m = 3,4$, is the continuous limit of a map defined by intersecting one $2$-dim plane and two $m-1$-dim subspaces in $\RP^m$. This pattern leads us to the following conjecture.

\begin{conjecture}
The AGD Hamiltonian flow associated to the $L^{\frac k{m+1}}$--Hamiltonian is the continuous limit of maps defined analogously to the pentagram map through the intersection of one $k-1$-dimensional subspace and $k-1$ $m-1$-dimensional subspaces of $\RP^m$.
\end{conjecture}

\begin{proof}[Proof of the theorem] The proof is a calculation similar to the one in the previous Theorem. Because of the symmetry in the integers $m_i$, $n_i$ and $r_i$, we will write the equations for the planes in terms of the elementary symmetric polynomials $M_i$, $N_i$ and $R_i$. This will both simplify the equations and will reduce their order, making it easier to solve.

If we use the notation in Theorem \ref{case2} we get similar initial equations, except for the fact that the plane corresponding to integers $m_1,m_2,m_3$ and associated to $\a = (a_1,a_2,a_3)$ coefficients is three dimensional, while the ones associated to integers $n_1,n_2,n_3,n_4$ and $r_1,r_2,r_3,r_4$ and associated to $\b$ and $\c$ are four dimensional, $\b=(b_1, b_2, b_3, b_4)$, $\c=(c_1,c_2,c_3,c_4)$. Therefore, instead of (\ref{1-2-3}) and subsequent equations, we have 
\begin{eqnarray*}\label{6-7-8}
1+\alpha_0\e+\beta_0\e^2+\g_0\e^3 +o(\e^3) &=& a_1+a_2+a_3+o(\e^3) \\&=& b_1+b_2+b_3+b_4+o(\e^3)\\&=& c_1+c_2+c_3+c_4+o(\e^3)\\
\alpha_1+\e\beta_1+\e^2\g_1+o(\e^2) &=& \m_1\cdot \a+o(\e^2) = \nb_1\cdot \b+o(\e^2)= \r_1\cdot \c +o(\e^2)\\
\be_2+\e\g_2+\e^2\d_2+o(\e^2) &=&\frac1{2!}\m_{2}\cdot \a+o(\e^2) \\&=& \frac1{2!}\nb_{2}\cdot \b + o(\e^2)=\frac1{2!}\r_{2}\cdot \c+o(\e^2)
\end{eqnarray*}
corresponding to the coefficients of $\l$, $\l'$ and $\l''$ and
\begin{eqnarray*}\label{4-5}
 \g_3+\e\d_3+\e^2\eta_3 + o(\e^2) &=& \frac1{3!} \m_{3}\cdot \a-\frac{\e^2}{5!}k_3 \m_{5}\cdot \a +o(\e^2)\\ &=& \frac1{3!} \nb_{3}\cdot \b-\frac{\e^2}{5!}k_3 \nb_{5}\cdot \b +o(\e^2)\\ &=& \frac1{3!} \r_{3}\cdot \c-\frac{\e^2}{5!}k_3 \r_{5}\cdot \c +o(\e^2)\\
 \d_4+\e\eta_4+\e^2\nu_4+o(\e^2) &=& \frac1{4!} \m_{4}\cdot \a-\frac{\e^2}{6!}k_3 \m_{6}\cdot \a + o(\e^2)\\ &=& \frac1{4!} \nb_{4}\cdot \b-\frac{\e^2}{6!}k_3 \nb_{6}\cdot \b + o(\e^2)\\ &=& \frac1{4!} \r_{4}\cdot \c-\frac{\e^2}{6!}k_3 \r_{6}\cdot \c + o(\e^2)
 \end{eqnarray*}
corresponding to the coefficients of $\l'''$ and $\l^{(4)}$. Using the first three equations we can solve for $\a_0$, $\a_1$ and $\a_2$ as $\a_i = A_3^{-1}(m)v_i$, where the subindex in the $A_3(m)$ refers to the size of the Vandermonde matrix, and where $v_i$ are as in (\ref{vi}). Using the fourth equation for $\a$ we can solve for $\g_3, \d_3$ and $\eta_3$, values that we will use in our next step. Using the first four equations for $\b$ and $\c$, we can solve for $\b_i$ and $\c_i$, $i = 0,1$ as $\b_i = A_4^{-1}(n)w_i$, $\c_i = A_4^{-1}(r)w_i$ with 
\[
w_0 = \begin{pmatrix}1\\\al_1\\2!\beta_2\\ 3!\gamma_3\end{pmatrix}, w_1 = \begin{pmatrix} \al_0\\\beta_1\\2!\g_2\\3!\d_3\end{pmatrix}
\]
and we can also solve for $\b_2 = A^{-1}_4(n)w_2^b$, $\c_2 = A^{-1}_4(r)w_2^c$ where
\[
w_2^b = \begin{pmatrix}\be_0\\\g_1\\2!\d_2\\3!\eta_3+\frac{3!}{5!}k_3\nb_5A^{-1}_4(n)w_0\end{pmatrix}, w_2^c = \begin{pmatrix}\be_0\\\g_1\\2!\d_2\\3!\eta_3+\frac{3!}{5!}k_3\r_5A^{-1}_4(r)w_0\end{pmatrix}.
\]
Substituting all these values in the last equation gives us a number of equations that will help us identify the parameter values for $\alpha_i, \beta_i, \gamma_i$, etc.

There are three zeroth order equations for $\d_4$. We can eliminate $\delta_4$ to obtain the following two equations for $\al_1$ and $\be_2$. (The calculations are a little long, but otherwise straightforward.)
\begin{eqnarray*}
M_1M_3-(N_1+N_4M_1)+[M_4-(N_2+N_4M_2)] \al_1+2![M_5-(N_3+N_4M_3)]\be_2 &=& 0\\
M_1M_3-(R_1+R_4M_1)+[M_4-(R_2+R_4M_2)] \al_1+2![M_5-(R_3+R_4M_3)]\be_2 &=& 0
\end{eqnarray*}
where $M_i$ $i=1,2,3$, $N_i$ and $R_i$, $i=1,2,3,4$ are as in lemma \ref{Vandermonde}, and $M_4 = M_3M_2+M_1$, $M_5 = M_3^2+M_2$ were given in the proof of the previous theorem. 
From here, conditions 
\begin{equation}\label{condition1}
M_1M_3 = N_1+N_4M_1,\hskip .2in M_1M_3 = R_1+R_4M_1
\end{equation}
together with the matrix
\begin{equation}\label{full rank}
M = \begin{pmatrix} M_4-(N_2+N_4M_2)&M_5-(N_3+N_4M_3)\\ M_4-(R_2+R_4M_2)&M_5-(R_3+R_4M_3)\end{pmatrix}
\end{equation}
having full rank, will ensure that $\alpha_1 = \beta_2 = 0$. As a result of this we have 
\[
v_0 = e_1,\hskip.2in \gamma_3= \frac 1{3!} M_1,\hskip.2in w_0 = e_1 + M_1 e_4.
\]
There are also three first order equations for $\eta_4$. Notice that $\al_0 = 0$ once we impose condition (\ref{norm}) to $\l_\e$. Again, after some rewriting we get the system 
\begin{eqnarray*}
0&=&[M_4-(N_2+N_4M_2)] \be_1+2![M_5-(N_3+N_4M_3)]\g_2  \\
0&=&[M_4-(R_2+R_4M_2)] \be_1+2![M_5-(R_3+R_4M_3)]\g_2 
\end{eqnarray*}
and hence the rank condition on matrix (\ref{full rank}) will ensure $\be_1 = \g_2 = 0$. With these values we also get $v_1 = 0$, $\d_3 = 0$ and $w_1=0$. Using the normalization of $\l_\e$ again we obtain $\be_0 = 0$, and from here 
\[
A = B = 0,\hskip 2ex C = \frac1{3!}\l''' + \g_1\l'+\g_0\l.
\]
We are left with the determination of $\g_1$, since $\g_0$ is determined by normalizing conditions.

From now on, let us assume that both the $n$-plane and the $r$-plane include $x_n$, so that we can assume that $n_4 = 0 = r_4$. In such case, $N_1 = R_1 = 0$ and so conditions (\ref{condition1}) becomes $M_3 = N_4 = R_4$. With this assumption, and using the fact that $M_4 = M_1+M_3M_2$ and $M_5 = M_3^2+M_2$, the matrix (\ref{full rank}) becomes
\begin{equation}\label{full rank2}
M=\begin{pmatrix} M_1-N_2&M_2-N_3\\ M_1-R_2&M_2-R_3\end{pmatrix}.
\end{equation}

Finally, we use the equations involving $\nu_4$. We have three of them, and so we can get rid of $\nu_4$ and obtain a system of equations for $\g_1$ and $\d_2$. The calculations are a little long, but they are straightforward if we use the following relations (themselves obtained straightforwardly).

\begin{eqnarray*}\m_6A_3^{-1}(m)e_1 &=& M_1(M_3^3+3M_1+2M_2M_3)\\ \nb_5A_4^{-1}(n)e_1 &=& N_1N_4\\ \nb_5A_4^{-1}(n)e_4 &=& N_4^2+N_3\\
 \nb_6A_4^{-1}(n)e_1 &=& N_1(N_4^2+N_3)\\ \nb_6A_4^{-1}(n)e_4 &=& N_4^3+2N_4N_3+N_2\end{eqnarray*}
and the fact that $N_4 = M_3$. The resulting system is given by $M \begin{pmatrix}\g_1\\ 2!\d_2\end{pmatrix} = N$, where $M$ is the matrix in (\ref{full rank2}) and where 
\[
N = -\frac1{60}k_3M_1\begin{pmatrix}M_3(N_3-M_2)+2(N_2-3M_1)\\ M_3(R_3-M_2)+2(R_2-3M_1)\end{pmatrix}.
\]
Clearly, $\g_1$ is then given by
\[
\g_1 = -\frac1{60\det M} k_3M_1\det\begin{pmatrix}M_3(N_3-M_2)+2(N_2-3M_1)&M_2-N_3\\ M_3(R_3-M_2)+2(R_2-3M_1)&M_2-R_3\end{pmatrix}.
\]
If we now impose the condition \[\g_1 = \frac{27}5\g_3 = \frac{27}{6!5} M_1 k_3\] we get the equation
\[
20\det\begin{pmatrix}M_3(N_3-M_2)+2(N_2-3M_1)&M_2-N_3\\ M_3(R_3-M_2)+2(R_2-3M_1)&M_2-R_3\end{pmatrix} = -9\det\begin{pmatrix}M_1-N_2&M_2-N_3\\ M_1-R_2&M_2-R_3\end{pmatrix}
\]
This equation can be easily programed. Using Maple to calculate the equations and a simple C-program to solve them, one can check that the smallest solutions involve integer values up to 7. Some of these solutions are
\[\begin{array}{ccc}
m_1 =  7, m_2 = -1, m_3 = -7,&\hskip .1in n_1 =  3, n_2 = -1, n_3 = -3,&\hskip .1in r_1 =  6, r_2 = -3, r_3 = -4,\\
m_1 =  7, m_2 = -1, m_3 = -7,&\hskip .1in n_1 =  3, n_2 = -1, n_3 = -3,&\hskip .1in r_1 =  4, r_2 = -2, r_3 = -3\\
m_1 =  7, m_2 = -1, m_3 = -7,&\hskip .1in n_1 =  6, n_2 = -3, n_3 = -4,&\hskip .1in r_1 =  4, r_2 = -2, r_3 = -3.
\end{array}
\]
\end{proof}
 

\end{document}